\documentclass{elsarticle}

\usepackage[english]{babel}
\usepackage[utf8]{inputenc} 
\usepackage[T1]{fontenc}
\usepackage{lmodern}

\usepackage{hyperref}
\usepackage{algorithm2e}
\usepackage{graphicx}
\usepackage{subfigure}
\usepackage{amssymb}
\usepackage{amsmath}
\usepackage{amsthm}
\usepackage{mathtools}
\usepackage{eqnarray}
\usepackage{amsfonts}
\usepackage{array}
\usepackage{tabularx}
\usepackage{multirow}
\usepackage{hyperref}
\usepackage{tikz}
\usepackage{changepage}
\usepackage[noend]{algorithmic}
\usepackage{enumitem}
\usepackage{arydshln}
\usetikzlibrary{shapes,shadows,arrows,positioning,graphs}

\newtheorem{lemma}{Lemma}
\newtheorem{hyp}{Hypothesis}
\newtheorem{defn}{Definition}
\theoremstyle{remark}
\newtheorem{rem}{Remark}
\newcommand{\argmin}{\arg\!\min}

\newcommand{\tagarray}{%
\mbox{}\refstepcounter{equation}%
$(\theequation)$%
}

\begin{document}
\begin{frontmatter}
\title{Exact algorithms for the order picking problem}

\author[myu]{Lucie Pansart\corref{cor1}}
\ead{lucie.pansart@g-scop.grenoble-inp.fr}

\author[myu]{Nicolas Catusse}
\ead{nicolas.catusse@g-scop.grenoble-inp.fr}

\author[myu]{Hadrien Cambazard}
\ead{hadrien.cambazard@g-scop.grenoble-inp.fr}

\address[myu]{Univ.Grenoble Alpes, CNRS, G-SCOP, 38000 Grenoble, France}

\cortext[cor1]{Corresponding author}

\begin{abstract}
Order picking is the problem of collecting a set of products in a warehouse in a minimum amount of time. It is currently a major bottleneck in supply-chain because of its cost in time and labor force. 
This article presents two exact and effective algorithms for this problem. Firstly, a sparse formulation in mixed-integer programming is strengthened by preprocessing and valid inequalities. Secondly, a dynamic programming approach generalizing known algorithms for two or three cross-aisles is proposed and evaluated experimentally. Performances of these algorithms are reported and compared with the Traveling Salesman Problem (TSP) solver Concorde.
\end{abstract}

\end{frontmatter}

\section{Introduction \label{sec:intro} } 

The order picking activity lies at the heart of logistic. It consists in collecting products from storage in a specific quantity given by a customer order. This process is often considered as the most important warehousing process since it is estimated that it accounts for 55\% of the total operational warehouse costs \cite{tompkins2010facilities}.
Many order picking systems are currently used in warehouses. The methods are classified following who picks the orders (humans or machines), who moves (picker or products) and the strategy used. We focus here on manual systems where the order picker moves in a regular rectangular warehouse.

An efficient way to optimize order picking in this case is to reduce picker travel time. Thus, we are concerned with the following issue: \textbf{how to optimize the routing time in the warehouse?}
This problem is called order picking problem, picker routing problem, or picking problem for sake of simplicity.
To the best of our knowledge, this problem cannot be solved in polynomial time, except when the number of cross-aisles is bounded \cite{roodbergen2001routingDP, cambazard2015fixed}.\\

We show in this paper that the problem can be efficiently solved optimally with mixed integer programming using a sparse formulation as well as adequate preprocessing and valid inequalities. Moreover, we extend a dynamic programming algorithm initially proposed by Ratliff and Rosenthal for 2 cross-aisles  \cite{ratliff1983order}, to any number of cross-aisles and evaluate it experimentally. It turns out to scale for a number of cross-aisles large enough to deal with real-life warehouses. This approach is however less suited to accommodate side-constraints.

This article begins with a description of the picking problem and a literature review (section \ref{sec:pb}). Section \ref{sec:pp} presents different ways of preprocessing the data to improve algorithms efficiency. Two exact algorithms are then presented in sections \ref{sec:milp} and \ref{sec:pdyn}. Finally, we report experimental results in section \ref{sec:results}.

\section{Problem description and literature review \label{sec:pb}}

We consider a regular rectangular warehouse with a single depot used to take the order and to drop it off. The warehouse is made of $v$  \textbf{vertical aisles} and $h$ \textbf{horizontal cross-aisles}. Products are located on both sides of vertical aisles. Cross-aisles do not contain any products but enable the order picker to navigate in the warehouse (see Figure \ref{fig:pb}). We make two common hypothesis on the warehouse:

\begin{hyp}
Aisles are narrow (the distance for crossing is considered null). \label{hyp:narrow}
\end{hyp}
\vspace{-2em}
\begin{hyp}
All aisle's lengths are equal. All cross-aisle's lengths are equal. \label{hyp:equi}
\end{hyp}

\begin{figure}[h]
\begin{adjustwidth}{-3cm}{-3cm}
\centering
\subfigure[Layout description of the considered warehouse ]{
\centering
\includegraphics[scale=0.58]{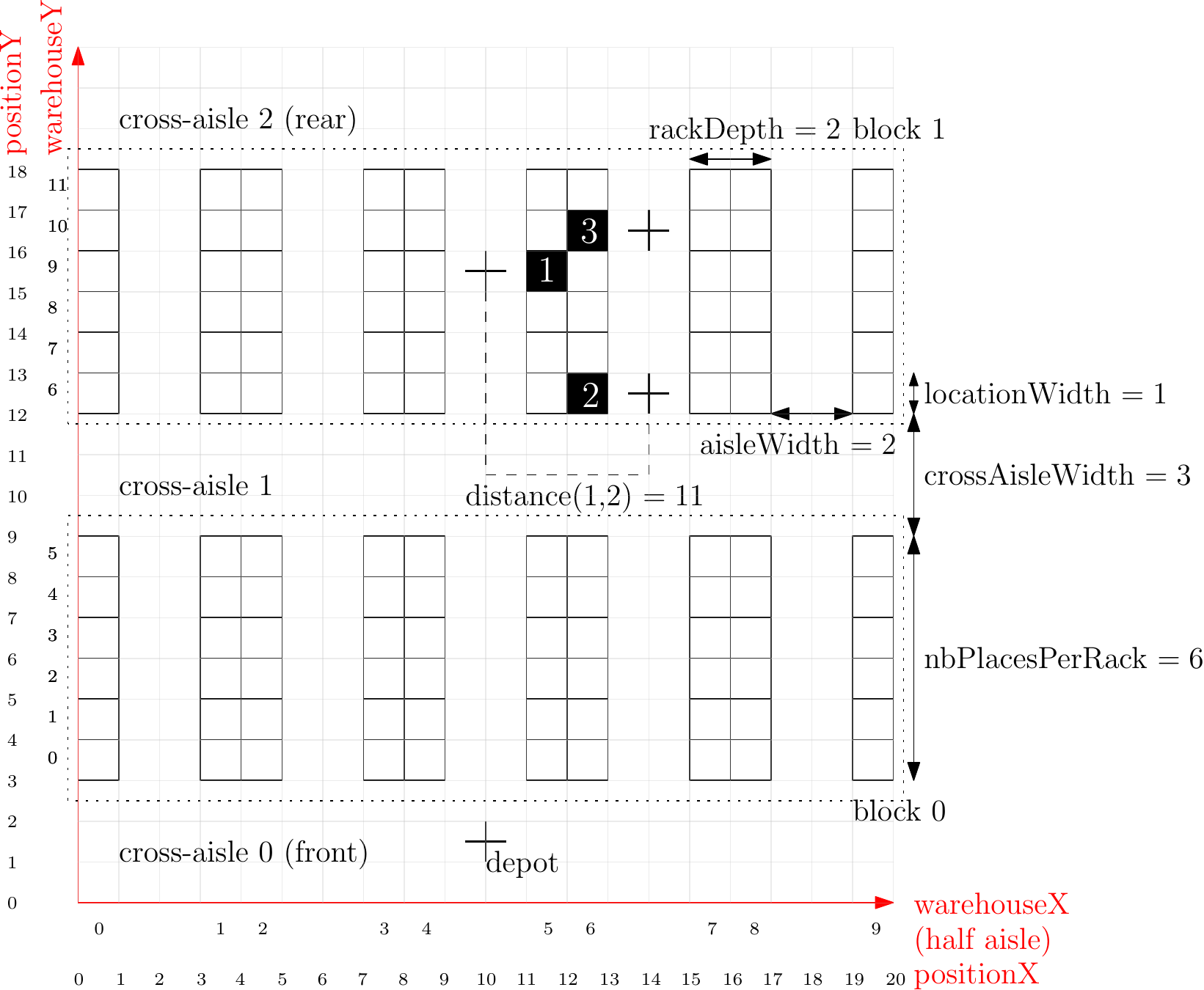}
 \label{fig:pb}}
\subfigure[Graph representation of the considered warehouse]{
\centering
\begin{tikzpicture}[thick,scale=0.75, every node/.style={scale=0.75}]
\node[opacity=0.4] (img) at (3,3.8){\includegraphics[scale=1.5]{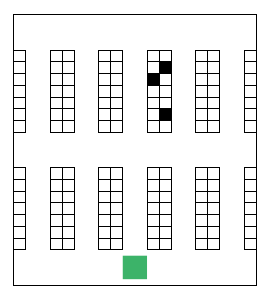}};

\tikzstyle{VR}=[circle,draw,fill=black,text=white,scale=0.8]
\tikzstyle{V}=[circle,draw,scale=0.8]
\tikzstyle{n}=[<->]
\tikzstyle{l}=[<->]

\node[V] (3)at(0,0){4};
\node[V] (4)at(1.5,0){5};
\node[V] (6)at(4.6,0){6};
\node[V] (7)at(6.2,0){7};
\node[V] (8)at(0,3.8){8};
\node[V] (9)at(1.5,3.8){9};
\node[V] (10)at(3,3.8){10};
\node[V] (11)at(4.6,3.8){11};
\node[V] (12)at(6.2,3.8){12};
\node[V] (13)at(0,7.4){13};
\node[V] (14)at(1.5,7.4){14};
\node[V] (15)at(3,7.4){15};
\node[V] (16)at(4.6,7.4){16};
\node[V] (17)at(6.2,7.4){17};
\node[VR] (d)at(3,0){d};
\node[VR] (1)at(3,6){1};
\node[VR] (2)at(4.6,4.9){2};
\node[VR] (p3)at(4.6,6.4){3};

\draw[n] (3) to (4);
\draw[n] (3) to (8);
\draw[n] (9) to (4);
\draw[n] (d) to (4);
\draw[n] (d) to (10);
\draw[n] (d) to (6);
\draw[n] (11) to (6);
\draw[n] (7) to (6);
\draw[n] (7) to (12);
\draw[n] (11) to (12);
\draw[n] (17) to (12);
\draw[n] (11) to (2);
\draw[n] (11) to (10);
\draw[n] (9) to (10);
\draw[n] (1) to (10);
\draw[n] (9) to (14);
\draw[n] (9) to (8);
\draw[n] (13) to (8);
\draw[n] (13) to (14);
\draw[n] (15) to (14);
\draw[n] (15) to (1);
\draw[n] (15) to (16);

\draw[n] (p3) to (16);
\draw[n] (2) to (p3);
\draw[n] (17) to (16);

\node (x) at (0,-1.7){};

\end{tikzpicture}
\label{fig:graph}}

\end{adjustwidth}
\caption{Example of a warehouse. A representation of the different parameters of a layout and the corresponding Steiner graph are overlaid on the warehouse structure.}
\end{figure}
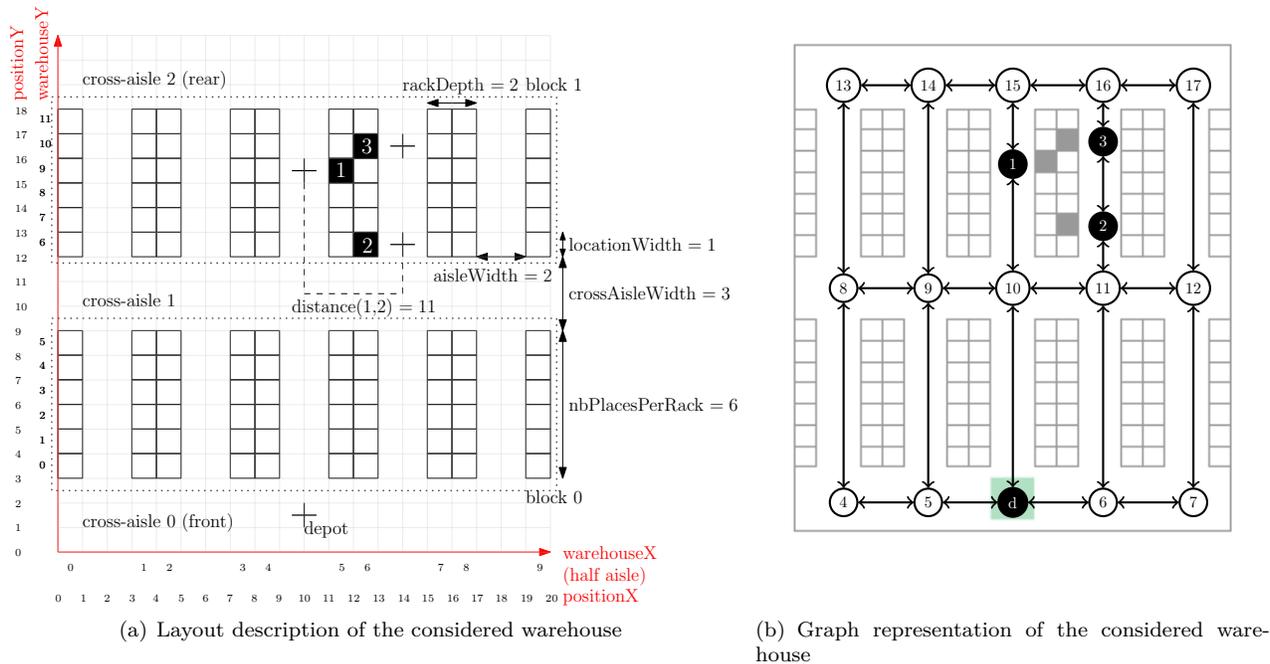

An \textbf{order} is given by a picking list, \emph{i.e.} a set of $n$ products, indexed from $1$ to $n$ and described by their location in the warehouse. 
We define by $\mathbf{R}$ the indexes of all locations to visit (the products) including $0$ which is the index of the depot so $R = \{0,\ldots,n\}$. The set $\mathbf{I}$ denotes the indexes of intersections between aisles and cross-aisles of the warehouse, excluding the depot. The set of all relevant locations in the warehouse is therefore $\mathbf{V} = I \cup R$. The problem is stated as follows: given $n$ products to pick in a rectangular warehouse, what is the shortest tour (beginning and ending at the depot) to collect all these products?

It is a particular case of the Traveling Salesman Problem (TSP), the well-known $\mathcal{NP}$-hard \cite{karp1972reducibility} problem, where the salesman is the order picker and cities are products to collect. This problem, introduced by Dantzig, Fulkerson and Johnson in 1954 \cite{dantzig1954solution}, is one of the most studied in Operations Research. The survey of Orman and Williams \cite{orman2006survey} gives an overview of integer programming formulations for the TSP.
Efficient exact algorithms have been designed for the TSP and Concorde is one of the best exact solver (see Hahsler and Hornik \cite{hahsler2007tsp} or  Mulder and Wunsch \cite{mulder2003million}).
To solve the specific case of picking problem, many heuristics have been proposed in particular by Hall \cite{hall1993distance}. Some performance analysis of the most popular heuristics were made by Petersen \cite{petersen1997evaluation} and by Roodbergen and De Koster \cite{roodbergen2001routingH}. Theys, Br{\"a}ysy, Dullaert and Raa \cite{theys2010using} proposed to combine classical TSP heuristics with picking heuristics and provided a benchmark which is used in this work.

An exact approach using dynamic programming has been proposed for the first time by Ratliff and Rosenthal in the case of a single block (i.e., 2 cross-aisles) in 1983 \cite{ratliff1983order}. It was extended by Roodbergen and De Koster \cite{roodbergen2001routingDP} in the case of three cross-aisles. These algorithms are polynomial in the number of aisles and products of the warehouse   and used in heuristics for cases with more than three cross-aisles \cite{vaughan1999effect,roodbergen2001routingH}. 
Cambazard and Catusse developed a dynamic programming approach which can solve any rectilinear TSP \cite{cambazard2015fixed}. This algorithm can be applied to the picking problem for any rectangular warehouse with $h$ cross-aisles, but it is exponential in $h$. 
 Although this extension was suggested in several articles \cite{roodbergen2001routingDP,ratliff1983order}, it has never been detailed nor implemented. The reason is that the extensive analysis made by \cite{ratliff1983order,roodbergen2001routingDP} of the possible states of the dynamic program does not easily generalize beyond three cross-aisles. We give in section \ref{sec:pdyn} a simple and general version of this exact algorithm for any number of cross-aisles to use it as a baseline in the experiments. \\ 

Practical comparisons have been performed between heuristics and exact algorithms, notably by De Koster and Van der Poort \cite{dekoster1998routing}. They compared the dynamic program with the commonly used S-shape heuristic and conclude that "the numerical results suggest that the savings in travel time may be substantial when using the optimal algorithm instead of the S-shape heuristic". This result strengthens the motivation for exact and efficient algorithms.


The problem can be seen as a Steiner TSP \cite{scholz2016new} which is a variant of the TSP proposed independently by Fleischmann \cite{fleischmann1985cutting} and Cornuéjols, Fonlupt and Naddef in 1985 \cite{cornuejols1985traveling} although Orloff introduced the idea some years before \cite{orloff1974fundamental}. Burkard \textit{et al.} \cite{burkard1998well} categorized the picking problem applied to the case of series-parallel graph in well-solvable special cases of the TSP, which includes warehouses with 2 cross-aisles \cite{cornuejols1985traveling}.
Our Mixed Integer Linear Programming (MILP) model is based on the compact MILP formulations  proposed by Letchord, Nasiri and Theis \cite{letchford2013compact}.\\

The Steiner TSP (also known as subset TSP) is stated in a directed graph $\mathbf{D=(V,A)}$ (see Figure \ref{fig:graph}) where $V$ is the set of vertices and $A$ is the set of arcs, defined as follows:
$\forall i,j \in V, \ (ij) \in A \Leftrightarrow $ one of the following conditions holds:
\begin{enumerate}
\item $i$ and $j$ are horizontally adjacent intersections (e.g. 4 and 5 in Figure \ref{fig:graph}).
\item $i$ and $j$ are extreme intersections of an empty sub-aisle (e.g. 4 and 8).
\item $j$ is an extreme products and $i$ the adjacent intersection (e.g. 3 and 16).
\item $i$ and $j$ are adjacent products (e.g. 2 and 3).
\end{enumerate}

Each arc $(ij)$ is weighted by  $d_{ij}$, the distance between $i$ and $j$ in the warehouse. We study here the symmetric TSP, which means the graph D is such that $\forall (ij) \in A : (ji) \in A$ and $d_{ij} = d_{ji}$. We generalize this distance by defining the function $d : V \times V \rightarrow \mathbb{N}$ as follows: $d(i,j)$ is the distance of a shortest path between vertices $i$ and $j$. 

We define by $P_{ij}$ the set of all shortest paths between the two products $i$ and $j$. More formally, $P_{ij} = \{ P=(v_0v_1v_2...v_K) | v_0 = i, \, v_K =j, \,(v_{k-1}v_{k})\in A$ \\ $ \forall k=1...K  \text{ and } \sum\limits_{k=1}^K  d_{v_{k-1}v_{k}} = d(i,j)$\}.
The set of \textbf{neighbors} of a vertex $i$ is denoted $\Gamma(i)$ so that $\Gamma(i)=\{ j \in V | (ij)  \in A\}$. Note that we don't need to distinguish the successeurs and predecessors of a vertex since they are the same by construction of $D$.
\section{Preprocessing \label{sec:pp}}

The size of the problem can be considerably reduced, which is a good lever to reduce computation time. First, we can reduce the size of the picking list, i.e., the number of vertices in $R$. Then, we show how to reduce the number of arcs in the graph, keeping a sufficient set of arcs to find an optimal solution. Figure \ref{fig:spanner} shows an example of the impact of the preprocessing.

\begin{figure}
\centering
\includegraphics[scale=0.15]{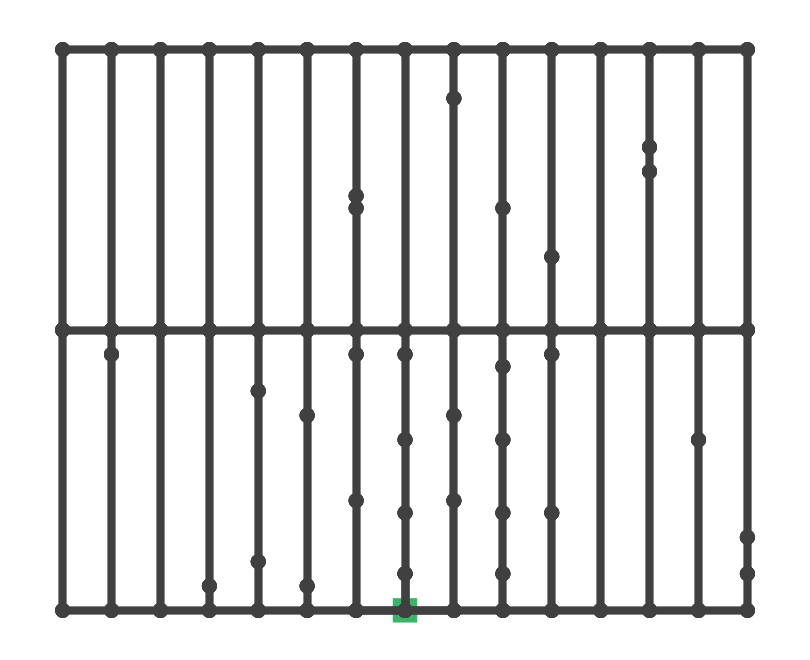}
\includegraphics[scale=0.15]{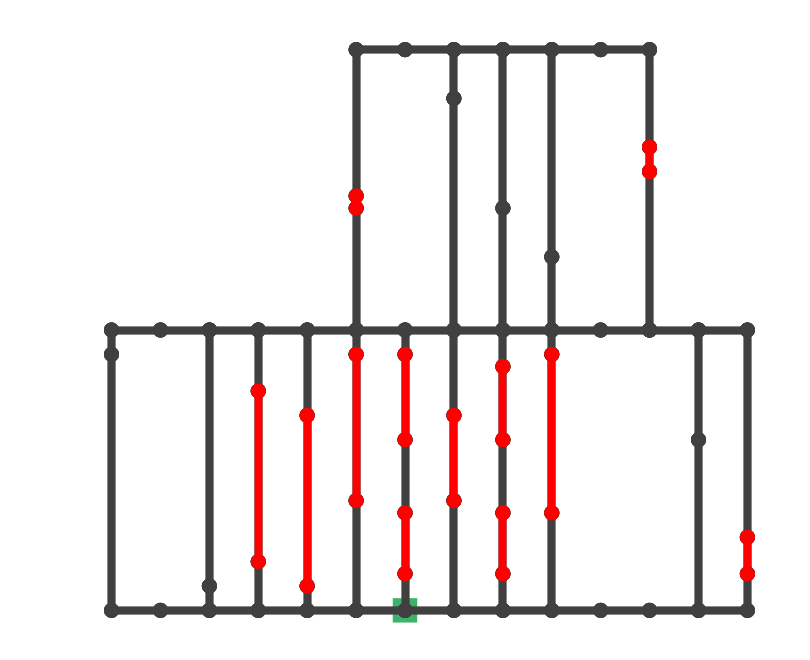}
\caption{Example of an instance and its preprocessing. Red edges show edges forced by the node preprocessing. The instance has 105 vertices and 270 arcs while the preprocessed instance has 65 vertices and 156 arcs.\label{fig:spanner}}
\end{figure}

\subsection{Reducing the number of vertices}

In this part, we describe two preprocessings reducing the number of products without changing the problem. The first algorithm must be used only when the solver accepts additional constraints while the second one is general.

Both algorithms are based on the following result due to Ratliff and Rosenthal: in an optimal solution, there exists six unique ways to traverse a sub-aisle (shown in Figure \ref{fig:6ways}) \cite{ratliff1983order}.
Case \hyperref[fig:caseE]{(e)} is the only configuration that might not be unique. However, only configurations where the gap between black and white groups is the largest will occur in an optimal solution. The largest gap of a sub-aisle is defined as the longest empty distance between two vertices of a sub-aisle. Several configurations can still exist if the largest gap is not unique in the sub-aisle, but they will all have the same cost in an optimal solution.

\begin{figure}
\centering
\subfigure[]{
\begin{tikzpicture}[thick,scale=0.5, every node/.style={transform shape}]
\tikzstyle{VR}=[circle,draw,color=black,scale=1.5]
\tikzstyle{VT}=[circle,draw,fill=black,scale=1.5]
\tikzstyle{V}=[circle,draw]
\tikzstyle{n}=[-]
\tikzstyle{d}=[<-]

\node[V] (i0)at(0,0){i0};
\node[V] (i1)at(0,9){i1};

\node[VT] (p1)at(0,1){};
\node[VT] (p7)at(0,1.75){};
\node[VT] (p2)at(0,2.5){};
\node[VR] (p3)at(0,5){};
\node[VR] (p6)at(0,5.5){};
\node[VR] (p4)at(0,6.5){};
\node[VR] (p5)at(0,8){};

\draw[n] (i0)--(p1);
\draw[n] (p1)--(p7);
\draw[n] (p7)--(p2);
\draw[n] (p2)--(p3);
\draw[n] (p3)--(p6);
\draw[n] (p6)--(p4);
\draw[n] (p4)--(p5);
\draw[d] (i1)--(p5);

\node at (0,-1){};
\end{tikzpicture}
\label{fig:caseA}}
\subfigure[]{
\begin{tikzpicture}[thick,scale=0.5, every node/.style={transform shape}]
\tikzstyle{VR}=[circle,draw,color=black,scale=1.5]
\tikzstyle{VT}=[circle,draw,fill=black,scale=1.5]
\tikzstyle{V}=[circle,draw]
\tikzstyle{n}=[-]
\tikzstyle{d}=[<-]

\node[V] (i0)at(0,0){i0};
\node[V] (i1)at(0,9){i1};

\node[VT] (p1)at(0,1){};
\node[VT] (p7)at(0,1.75){};
\node[VT] (p2)at(0,2.5){};
\node[VR] (p3)at(0,5){};
\node[VR] (p6)at(0,5.5){};
\node[VR] (p4)at(0,6.5){};
\node[VR] (p5)at(0,8){};

\draw[d] (i0)--(p1);
\draw[n] (p1)--(p7);
\draw[n] (p7)--(p2);
\draw[n] (p2)--(p3);
\draw[n] (p3)--(p6);
\draw[n] (p6)--(p4);
\draw[n] (p4)--(p5);
\draw[n] (i1)--(p5);

\node at (0,-1){};
\end{tikzpicture}
\label{fig:caseB}}
\subfigure[]{
\begin{tikzpicture}[thick,scale=0.5, every node/.style={transform shape}]
\tikzstyle{VR}=[circle,draw,color=black,scale=1.5]
\tikzstyle{VT}=[circle,draw,fill=black,scale=1.5]
\tikzstyle{V}=[circle,draw]
\tikzstyle{I}=[circle,scale=1.4]
\tikzstyle{n}=[-]
\tikzstyle{d} = [<-]
\tikzstyle{e} = [->]

\node[V] (i0)at(0,0){i0};
\node[V] (i1)at(0,9){i1};
\node[I] (ii0)at(0,0.2){};
\node[I] (ii1)at(0,8.8){};

\node[VT] (p1)at(0,1){};
\node[VT] (p7)at(0,1.75){};
\node[VT] (p2)at(0,2.5){};
\node[VR] (p3)at(0,5){};
\node[VR] (p6)at(0,5.5){};
\node[VR] (p4)at(0,6.5){};
\node[VR] (p5)at(0,8){};

\draw[e] (p2.east)--(p1.east);
\draw[n] (p2.east)--(p3.east);
\draw[n] (p4.east)--(p3.east);
\draw[n] (p4.east)--(p5.east);
\draw[n] (ii1.east)--(p5.east);

\draw[n] (p2.west)--(p1.west);
\draw[n] (p2.west)--(p3.west);
\draw[n] (p4.west)--(p3.west);
\draw[n] (p4.west)--(p5.west);
\draw[d] (ii1.west)--(p5.west);

\node at (0,-1){};
\end{tikzpicture}
\label{fig:caseC}}
\subfigure[]{
\begin{tikzpicture}[thick,scale=0.5, every node/.style={transform shape}]
\tikzstyle{VR}=[circle,draw,color=black,scale=1.5]
\tikzstyle{VT}=[circle,draw,fill=black,scale=1.5]
\tikzstyle{V}=[circle,draw]
\tikzstyle{I}=[circle,scale=1.4]
\tikzstyle{n}=[-]
\tikzstyle{d} = [<-]
\tikzstyle{e} = [->]

\node[V] (i0)at(0,0){i0};
\node[V] (i1)at(0,9){i1};
\node[I] (ii0)at(0,0.2){};
\node[I] (ii1)at(0,8.8){};

\node[VT] (p1)at(0,1){};
\node[VT] (p7)at(0,1.75){};
\node[VT] (p2)at(0,2.5){};
\node[VR] (p3)at(0,5){};
\node[VR] (p6)at(0,5.5){};
\node[VR] (p4)at(0,6.5){};
\node[VR] (p5)at(0,8){};

\draw[d] (ii0.east)--(p1.east);
\draw[n] (p2.east)--(p1.east);
\draw[n] (p2.east)--(p3.east);
\draw[n] (p4.east)--(p3.east);
\draw[n] (p4.east)--(p5.east);

\draw[n] (ii0.west)--(p1.west);
\draw[n] (p2.west)--(p1.west);
\draw[n] (p2.west)--(p3.west);
\draw[n] (p4.west)--(p3.west);
\draw[e] (p4.west)--(p5.west);

\node at (0,-1){};
\end{tikzpicture}
\label{fig:caseD}}
\subfigure[]{
\begin{tikzpicture}[thick,scale=0.5, every node/.style={transform shape}]
\tikzstyle{VT}=[circle,draw,fill=black,scale=1.5]
\tikzstyle{VR}=[circle,draw,color=black,scale=1.5]
\tikzstyle{V}=[circle,draw]
\tikzstyle{I}=[circle,scale=1.4]
\tikzstyle{n}=[-]
\tikzstyle{d} = [<-]
\tikzstyle{e} = [->]

\node[V] (i0)at(0,0){i0};
\node[V] (i1)at(0,9){i1};
\node[I] (ii0)at(0,0.2){};
\node[I] (ii1)at(0,8.8){};

\node[VT] (p1)at(0,1){};
\node[VT] (p7)at(0,1.75){};
\node[VT] (p2)at(0,2.5){};
\node[VR] (p3)at(0,5){};
\node[VR] (p6)at(0,5.5){};
\node[VR] (p4)at(0,6.5){};
\node[VR] (p5)at(0,8){};

\draw[d] (ii0.east)--(p1.east);
\draw[n] (p2.east)--(p1.east);

\draw[n] (p4.east)--(p3.east);
\draw[n] (p4.east)--(p5.east);
\draw[d] (ii1.east)--(p5.east);

\draw[n] (ii0.west)--(p1.west);
\draw[d] (p2.west)--(p1.west);
\draw[e] (p4.west)--(p3.west);

\draw[n] (p4.west)--(p5.west);
\draw[n] (ii1.west)--(p5.west);

\node at (0,-1){};
\end{tikzpicture}
\label{fig:caseE}}
\subfigure[]{
\begin{tikzpicture}[thick,scale=0.5, every node/.style={transform shape}]
\tikzstyle{VT}=[circle,draw,fill=black,scale=1.5]
\tikzstyle{VR}=[circle,draw,color=black,scale=1.5]
\tikzstyle{V}=[circle,draw]
\tikzstyle{I}=[circle,scale=1.4]
\tikzstyle{n}=[-]
\tikzstyle{d} = [<-]
\tikzstyle{e} = [->]

\node[V] (i0)at(0,0){i0};
\node[V] (i1)at(0,9){i1};
\node[I] (ii0)at(0,0.2){};
\node[I] (ii1)at(0,8.8){};

\node[VT] (p1)at(0,1){};
\node[VT] (p7)at(0,1.75){};
\node[VT] (p2)at(0,2.5){};
\node[VR] (p3)at(0,5){};
\node[VR] (p6)at(0,5.5){};
\node[VR] (p4)at(0,6.5){};
\node[VR] (p5)at(0,8){};

\draw[d] (ii0.east)--(p1.east);
\draw[n] (p2.east)--(p1.east);
\draw[n] (p2.east)--(p3.east);
\draw[n] (p4.east)--(p3.east);
\draw[n] (p4.east)--(p5.east);
\draw[n] (ii1.east)--(p5.east);

\draw[n] (ii0.west)--(p1.west);
\draw[n] (p2.west)--(p1.west);
\draw[n] (p4.west)--(p3.west);
\draw[n] (p2.west)--(p3.west);
\draw[n] (p4.west)--(p5.west);
\draw[d] (ii1.west)--(p5.west);

\node at (0,-1){};
\end{tikzpicture}
\label{fig:caseF}}
\subfigure[]{
\begin{tikzpicture}[thick,scale=0.5, every node/.style={transform shape}]

\tikzstyle{VR}=[circle,draw,color=black,scale=1.5]
\tikzstyle{VT}=[circle,draw,fill=black,scale=1.5]
\tikzstyle{NR}=[circle,draw,dotted,color=black,scale=1.5]
\tikzstyle{NT}=[circle,draw,dotted,fill=black!20,scale=1.5]
\tikzstyle{V}=[circle,draw]
\tikzstyle{n}=[-]

\node[V] (i0)at(0,0){i0};
\node[V] (i1)at(0,9){i1};

\node[VT] (p1)at(0,1){} ;
\node at (0.5,1) {$b_S$};

\node[NT] (p7)at(0,1.75){};
\node[VT] (p2)at(0,2.5){};
\node at (0.5,2.5) {$t_S$};

\node[VR] (p3)at(0,5){};
\node at (0.5,5) {$b_T$};
\node[NR] (p6)at(0,5.5){};
\node[NR] (p4)at(0,6.5){};
\node[VR] (p5)at(0,8){};
\node at (0.5,8) {$t_T$};
\node at (-4,0){};
\node at (4,0){};

\draw[n] (p1)--(p2);

\draw[n] (p3)--(p5);

\draw[n] (-2,-1)--(-2,10);
\draw[n] (2,-1)--(2,10);
\end{tikzpicture}
\label{fig:6waysPP}}
\subfigure[]{
\begin{tikzpicture}[thick,scale=0.5, every node/.style={transform shape}]

\tikzstyle{VR}=[circle,draw,color=black,scale=1.5]
\tikzstyle{VT}=[circle,draw,fill=black,scale=1.5]
\tikzstyle{NR}=[circle,draw,dotted,color=black,scale=1.5]
\tikzstyle{NT}=[circle,draw,dotted,fill=black!20,scale=1.5]
\tikzstyle{V}=[circle,draw]
\tikzstyle{n}=[-]

\node[V] (i0)at(0,0){i0};
\node[V] (i1)at(0,9){i1};

\node[VT] (p1)at(0,1){} ;

\node[NT] (p7)at(0,1.75){};
\node[VT] (p2)at(0,2.5){};

\node[VR] (p3)at(0,5){};

\node[NR] (p6)at(0,5.5){};
\node[VR] (p4)at(0,6.5){};
\node[VR] (p5)at(0,8){};

\node at (0,-1){};
\end{tikzpicture}
\label{fig:6waysPPC}}
\caption{To the left: The six ways to traverse a sub-aisle. We distinguish two groups (black and white vertices) where all vertices of one group are taken in one go. In the middle: the result after preprocessing \ref{sec:ppbasic}. Dashed vertices are removed from the problem and lines show mandatory paths.
To the right: the result after preprocessing \ref{sec:ppconcorde} \label{fig:6ways}}
\end{figure}
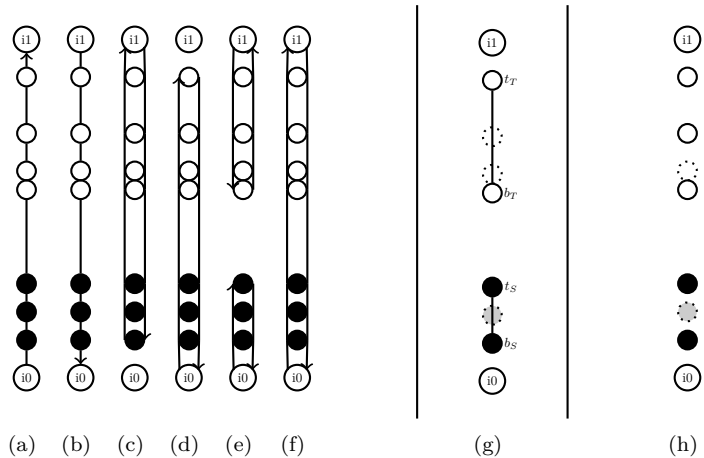

\subsubsection{Additional constraints available \label{sec:ppbasic}}

In any case, the black vertices are all picked in one go (one after the other) and the white vertices as well. So, for both subsets, we can keep only extreme products and impose an arc to be taken between these two extreme products.

\begin{defn}[Preprocessing]
The vertex preprocessing is defined by the following algorithm: \\
\begin{algorithm}[H]
\For{ every sub-aisle} {
\item[-] Compute a largest gap between two vertices
\item[-] Identify the set S (resp. T) containing all products below (resp. above) the largest gap
\item[-] In each subset (S and T), keep the two products $t_S$ and $b_S$ (resp. $t_T$ and $b_T$) that are the farthest apart
\item[-] Add the constraints forcing the order picker to traverse each set at least once. This means the order picker must traverse arc $(t_S b_S)$ or $(b_S t_S)$ and similarly arc  $(t_T b_T)$ or $(b_T t_T)$. \protect\footnotemark}
\end{algorithm} 
\end{defn}
\footnotetext{In the MILP (see section \ref{sec:milp}) it translates into the following constraints:\\
$$ x_{t_S b_S} + x_{b_S t_S} \geq 1 \text{ if } t_S \neq b_S \text{ and } x_{t_T b_T} + x_{b_T t_T} \geq 1 \text{ if } t_T \neq b_T \label{eq:cstrs} $$}

Note that $S$ and $T$ can be empty or singletons ($t_S = b_S$). Figure \ref{fig:6waysPP} shows the result after preprocessing the sub-aisle depicted. The vertices of set $S$ are drawn in black whereas the vertices of $T$ are white. Note that the six configurations have either $(t_S b_S)$ or $(b_S t_S)$  (or both) on one side and either $(t_T b_T)$ or $(b_T t_T)$ (or both) on the other side, hence the constraints added by the preprocessing.
\begin{rem}
The preprocessing does not change the value of the optimal solution and leaves \textbf{at most 4 products} by sub-aisle.
\end{rem}

A similar preprocessing was proposed independently by Scholz \textit{et al.} \cite{scholz2016new}, using different constraints to force all the products of one set (S or T) to be picked together.

\subsubsection{With a distance matrix \label{sec:ppconcorde}}

When it is not allowed to add constraints, for example when using Concorde, we can still suppress some vertices. Indeed, the constraints \eqref{eq:cstrs} added by the preprocessing are required only if the removal of vertices changes the position of the largest gap in the sub-aisle. Thus, we can still suppress vertices as long as the largest gap remains between $b_T$ and $t_S$ (see Figure \ref{fig:6waysPPC}).

\subsection{Reducing the number of arcs}

Let's now focus on the arcs. We compute the minimum 1-spanner in terms of number of edges, i.e., we are looking for a subset of the arcs preserving at least one original shortest path between any pairs of required vertices. The rationale is that any solution of the picking problem gives a tour where two products picked successively are linked by a shortest path. So any optimal solution can be found in a 1-spanner of $G$.

\begin{defn}
 A \textbf{k-spanner} of a graph $G$ is a sub-graph $H \subset G$ such that: 
\begin{itemize}[noitemsep,topsep=0pt]
\setlength\itemsep{1pt}
\item $H$ contains all the vertices of $G$. 
\item The distance between each pair of vertices in $H$ is at most $k$ times their distance in $G$.
\end{itemize}
\end{defn}

In our case, we search for a minimum "1-spanner" in the Steiner graph, restricted to the required vertices, which can also be seen as a particular case of the Minimum Manhattan Network. Let's consider the undirected graph $G = (V = I \cup R , E)$ where $E$ is defined as $A$ by removing orientation. We search for a subgraph of $G$ such that between each pair $(i,j)$ of required vertices there exists a shortest path which keeps the initial distance $d(i,j)$. Namely we search a graph $H = (V_H, E_H) $ such that: $R \subset  V_H  $, $V_H \subset V$, $E_H \subset E$ and  $\forall i,j \in V_H : \exists $ a $(i,j)$-path $P \in H |  \sum\limits_{(uv)\in P}  d_{uv} = d(i,j)$.\\

To compute a minimum 1-spanner, we choose the minimum set of edges with respect to the preceding properties.
 We use an integer linear program to solve the problem of the \textbf{1-spanner}, based on a model for the Minimum Manhattan Network from Benkert \textit{et al.} (\cite{benkert2006minimum}). The resulting set of selected edges $E^*$ forms the 1-spanner which can be used to state the Steiner TSP.
 
The Minimum Manhattan Network is $\mathcal{NP}$-hard \cite{chin2011minimum} but its resolution with MILP turns out to be really fast in practice for the benchmark considered. Finally, note that we don't need to compute the optimal network and any feasible 1-spanner can be used for preprocessing. 



\section{MILP formulation \label{sec:milp}}

The Steiner variant of the TSP was proposed by Cornuéjols, Fonlupt and Naddef in 1985 \cite{cornuejols1985traveling}. It was introduced especially to solve problems where the graph is sparse. The principle is that the graph contains some \textbf{required vertices} which \textbf{must} be visited and some \textbf{Steiner vertices} which \textbf{can} be visited. Moreover, in a Steiner Traveling Salesman Problem (STSP), the graph is not complete and edges as well as \textbf{vertices can be visited more than once}. In this section, we apply a Steiner approach to the picking problem.

We define the Steiner graph from the directed graph modeling an instance (see Figure \ref{fig:graph}): the required vertices are the products plus the depot and the Steiner vertices are the intersections in the warehouse. We look for a shortest tour in this graph, going through all the required vertices \textbf{at least} once.

\subsection{Flow-based formulation \label{ssec:scfs}}

We can use the compact single commodity flow formulation proposed by Letchford, Nasiri and Theis \cite{letchford2013compact}. It follows the flow principle: the order picker leaves the depot with $n$ units of a commodity and delivers one unit each time he picks a product. \\
 
We define the variables: $\forall (ij) \in A $ \\
$ x_{ij} = \left\{ 
\begin{aligned} &1 \text{ if the tour uses the arc $(ij)$} \\ 
&0 \text{ otherwise} \\
\end{aligned} \right. $\\
$ y_{ij} =  $ amount of commodity passing through arc $(ij)$.\\

The solution of the STSP is then found by solving the following mixed-integer linear program further called \textbf{SCFS}, standing for Single Commodity Flow Formulation for a Steiner TSP: \label{SCFS}\\

\hspace{-1cm}
\begin{tabularx}{\linewidth}{crllcr}
\multirow{8}*{$(SCFS)$} & $z^* = min \quad $& $\sum\limits_{(ij)\in A} d_{ij}x_{ij} $& &\tagarray\label{eq:SCFS_obj}  & \\
& s.t. & $\sum\limits_{j\in \Gamma(i)} x_{ij} \geq 1$ & $\forall i \in R $ & \tagarray\label{eq:SCFS_ass1}& 
 \\
 && $\sum\limits_{j\in \Gamma(i)} x_{ij} = \sum\limits_{j\in \Gamma(i)} x_{ji} \quad$  &$\forall i \in V $ &  \tagarray\label{eq:SCFS_ass2} &
 \\
 && $\sum\limits_{j\in \Gamma(i)} y_{ji} - \sum\limits_{j\in \Gamma(i)} y_{ij} = 1 \quad$ &$\forall i \in R\setminus\{0\}  $ &  \tagarray\label{eq:SCFS_flowR} & 
 \\
 && $\sum\limits_{j\in \Gamma(i)} y_{ji} - \sum\limits_{j\in \Gamma(i)} y_{ij} = 0 \quad$ &$\forall i \in V\setminus\{R\} $ &   \tagarray\label{eq:SCFS_flowS} & 
 \\ 
 &&$  y_{ij} \leq n x_{ij}$ & $\forall (ij) \in A  $ & \tagarray\label{eq:SCFS_ybound}& 
 \\
&& $x_{ij} \in \mathbb{N}$ &$\forall (ij) \in A  $ & \tagarray\label{eq:SCFS_xdomain}& \\
&& $y_{ij} \geq 0$ &$\forall (ij) \in A  $ & \tagarray\label{eq:SCFS_ydomain}& 
\end{tabularx}
$ $ \\

Constraints \eqref{eq:SCFS_ass1} ensure that each required vertex is visited \textbf{at least} once. Constraints \eqref{eq:SCFS_ass2} ensure that the tour arrives in any vertex as many times as it leaves it. 
The flow constraints are different depending on whether a vertex is required or not: constraints \eqref{eq:SCFS_flowR} impose that the order picker delivers one unit of the commodity to each product while constraints \eqref{eq:SCFS_flowS} impose that the flow stays the same through a non-required vertex.
Constraints \eqref{eq:SCFS_ybound} link the $y$ to the $x$ variables so that if some flow transits through $(ij)$ then the arc $(ij)$ is chosen.

\begin{rem}
The variables $y$ are real variables (8) but the optimal solution will be integer due to the fact that $y$ represent a flow.
\end{rem}
\begin{rem}
	We know from Lemma 1 in Letchford, Nasiri and Theis \cite{letchford2013compact} that every optimal solution of the STSP uses an arc at most once and thus satisfies $x_{ij}\leq 1$, $\forall (ij) \in A $. It is therefore sufficient to define $x$ as a positive integer (constraints \eqref{eq:SCFS_xdomain}). So the linear relaxation amounts to $x_{ij} \geq 0$.
\end{rem}

This formulation is compact and really sparse, especially thanks to the preprocessings described above. However, we notice that the quality of the lower bound given by the linear relaxation is weaker than the one given by a more standard formulation as it is explained below.

\subsection{Theoretical study of the formulation}

In this section, we compare the formulation described above with a standard flow-based formulation for the TSP. We consider the problem of finding a shortest tour in the complete and directed graph composed of all the products and where the distance between two vertices is the shortest distance in the warehouse. For sake of simplicity, we use the directed version of the TSP. Thus, we compare the standard single-commodity flow formulation for the TSP (as defined by Gavish \& Graves \cite{gavish1978travelling}) with the single-commodity flow formulation described above for the Steiner TSP. 

We recall the Gavish and Graves formulation (denoted by \textbf{SCF}), adapted to a directed graph:

\begin{tabularx}{\linewidth}{crllcr}
\multirow{8}*{$(SCF)$} & $z'^* = min \quad $& $\sum\limits_{\substack{i,j\in R\\i\neq j}} d(i,j)x'_{ij}$& &\tagarray\label{eq:SCF_obj}  & \\
& s.t. & $\sum\limits_{\substack{j\in R\\j\neq i}} x'_{ij} = 1$ & $\forall i \in R $ & \tagarray\label{eq:SCF_ass1}&
\\
 && $\sum\limits_{\substack{j\in R\\j\neq i}} x'_{ji} = 1$  &$\forall i \in R $ &  \tagarray\label{eq:SCF_ass2} &  \\
 && $\sum\limits_{\substack{j\in R\\j\neq i}} y'_{ji} - \sum\limits_{\substack{j\in R\\j\neq i}}y'_{ij} = 1 \quad$ &$\forall i \in R\setminus\{0\}  $ &  \tagarray\label{eq:SCF_flow} & 
 \\ 
 &&$  y'_{ij} \leq n x'_{ij}$ & $\forall i,j \in R, \ i\neq j  $ & \tagarray\label{eq:SCF_ybound}& 
 \\
&& $x'_{ij} \in \{0,1\}$ &$\forall i,j \in R, \ i\neq j  $ & \tagarray\label{eq:SCF_xdomain}& \\
&& $y'_{ij} \geq 0$ &$\forall i,j \in R , \ i\neq j $ & \tagarray\label{eq:SCF_ydomain}& \\
\end{tabularx}
$ $ \\

Constraints \eqref{eq:SCF_ass1} and \eqref{eq:SCF_ass2} are usual assignment constraints ensuring that each vertex is visited \textbf{exactly} once. 
Constraints \eqref{eq:SCF_flow} ensure that, except for the depot, the salesman deliver one unit at each vertex and retains the rest of the flow.
Constraints \eqref{eq:SCF_ybound} are \textit{Big-M} constraints linking $y$ and $x$. 
Finally, the objective \eqref{eq:SCF_obj} is to minimize the total distance of the tour.

To compare both formulations, we define a projection $proj$ which projects a fractional solution of (SCF) in the space of (SCFS) by keeping the distance value. The idea is to divide the value $x'$ between two required vertices on all the shortest paths.

\begin{defn}[Projection $proj$]
	\begin{equation}
		\begin{aligned}
			proj: &   & \mathbb{R}^{|R|^2} \rightarrow & \mathbb{R}^{|A|} &   \\
			&   & x' \rightarrow                 & x                
		\end{aligned}
	\end{equation}
				
	Where $proj (x')$ is defined by: 
	\begin{equation}
		x_{uv} = \sum\limits_{i,j\in R}\sum\limits_{\substack{P\in P_{ij} \\  (uv) \in P}} \frac{x'_{ij}}{|P_{ij}|} \ \forall (uv) \in A
	\end{equation}
\end{defn}

By construction, the conservation is checked at each vertex (property \eqref{eq:P1_2}) and each required vertex is "visited" at least as many times as in the standard TSP solution (property \eqref{eq:P1_1}). A required vertex can be visited more if it lies on a shortest path between two other required vertices. 
\begin{rem}[Properties of $proj$]
	Let $x' \in \mathbb{R}^{|R|^2}$ and $x = proj(x')$. Then:  
		      \begin{equation} 
		      (i) \:\: \sum\limits_{j\in \Gamma(i)}x_{ij} \geq \sum\limits_{j\in R} x'_{ij} \ \forall i \in R \ \label{eq:P1_1}
		      \end{equation}
		      		      		      		      
		\begin{equation} 
		(ii) \:\:  \sum\limits_{j\in \Gamma(i)}x_{ij} = \sum\limits_{j\in \Gamma(i)}x_{ji} \ \forall i \in V \label{eq:P1_2}
		\end{equation}
		
		\begin{equation*}
		 (iii) \:\: \sum\limits_{(ij)\in A}x_{ij}d_{ij} = \sum\limits_{i,j\in R}x'_{ij}d(i,j)
		\end{equation*}
\end{rem}

\begin{lemma}
	The linear relaxation of the Steiner single commodity flow formulation (SCFS) is \textbf{weaker} than the linear relaxation of the standard TSP single commodity flow formulation (SCF).
\end{lemma}
\begin{proof}$  $
Every projection of an optimal solution of SCF is feasible in SCFS, so $z_{LP}'^* \geq z_{LP}^*$ and there exists a case for which $z_{LP}'^* \neq z_{LP}^*$ so that SCF is a better formulation than SCFS. The two cases are detailed below:
	\begin{itemize}
		\item[$\bullet$]$z_{LP}^* \leq z_{LP}'^* $
	\end{itemize}
	 Let $(x',y')$ an optimal fractional solution of the linear relaxation of (SCF) and $(x,y) = (proj(x'),proj(y'))$. Then $(x,y)$ is a feasible solution for the linear relaxation of  (SCFS) since:
	 \begin{itemize}[noitemsep,topsep=0pt]
	 \setlength\itemsep{1pt}
	 \item[-] Assignment constraints \eqref{eq:SCFS_ass1} are respected due to the property \eqref{eq:P1_1} of $proj$.
	 \item[-] Conservation constraints are respected due to the property \eqref{eq:P1_2} of $proj$.
	 \item[-] For the same reasons, the flow constraints \eqref{eq:SCFS_flowR} and \eqref{eq:SCFS_flowS} are respected.
	 \item[-] Constraints \eqref{eq:SCFS_ybound} is respected: since the transformation is the same on $x'$ and $y'$ to obtain $x$ and $y$, the ratio is kept: $ \frac{y_{uv}}{x_{uv}} = \frac{y'_{ij}}{x'_{ij}} \leq n$ so \\ $ y_{uv} \leq nx_{uv}\ \forall (uv) \in P_{ij}, \forall i,j \in R$
					
	 \end{itemize}
	Thus, $(x,y)$ is feasible for $SCFS$ and $z_{LP}=z_{LP}'$.
	\begin{itemize}
		\item[$\bullet$]$z_{LP}^* \neq z_{LP}'^* $
	\end{itemize}
	
We consider an example with 3 products and show that $z_{LP}^*\leq 18 < z_{LP}'^* = 20$. Figure \ref{fig:representations} shows the representations of this example with the Steiner graph and with the complete graph. 
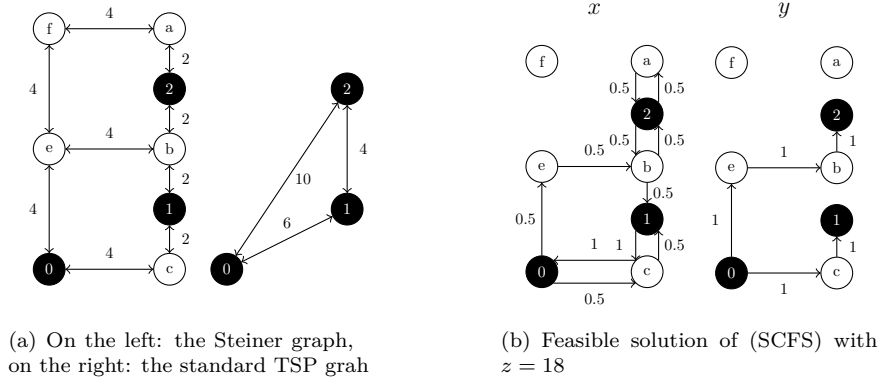
\begin{figure}
\begin{center}
\subfigure[On the left: the Steiner graph, \quad \quad \quad \quad \quad \ \ $ $ \ on the right: the standard TSP grah \label{fig:representations}]{
\begin{tikzpicture}[scale=0.8,every node/.style={scale=0.6,minimum size=20pt}]
\tikzstyle{VR}=[circle,draw,fill=black,text=white]
\tikzstyle{V}=[circle,draw]
\tikzstyle{n} = [<->]

\node[VR] (1) at (2,1){1};
\node[VR] (2) at (2,3){2};
\node[V] (f) at (0,4){f};
\node[V] (a) at (2,4){a};
\node[V] (e) at (0,2){e};
\node[V] (b) at (2,2){b};
\node[V] (c) at (2,0){c};
\node[VR] (0) at (0,0){0};

\draw[n] (0)--(e) node[midway, left] {4};
\draw[n] (f)--(e) node[midway, left] {4};
\draw[n] (a)--(2) node[midway, right] {2};
\draw[n] (b)--(2) node[midway, right] {2};
\draw[n] (1)--(b) node[midway, right] {2};
\draw[n] (c)--(1) node[midway, right] {2};
\draw[n] (f)--(a) node[midway, above] {4};
\draw[n] (e)--(b) node[midway, above] {4};
\draw[n] (0)--(c) node[midway, above] {4};

\node () at (0,-0.5) {};
\end{tikzpicture}
\begin{tikzpicture}[scale=0.8,every node/.style={scale=0.6,minimum size=20pt}]
\tikzstyle{VR}=[circle,draw,fill=black,text=white]
\tikzstyle{n} = [<->]

\node[VR] (1) at (2,1){1};
\node[VR] (2) at (2,3){2};
\node[VR] (d) at (0,0){0};

\node () at (4,-0.5) {};
\draw[n] (d)--(1) node[midway, above] {6};
\draw[n] (d)--(2) node[midway, right] {10};
\draw[n] (2)--(1) node[midway, right] {4};
\end{tikzpicture}
}
\subfigure[Feasible solution of (SCFS) with  $z=18$ \label{fig:feassol} ]{
\begin{tikzpicture}[scale=0.7,every node/.style={scale=0.6,minimum size=20pt}]
\tikzstyle{VR}=[circle,draw,fill=black,text=white]
\tikzstyle{V}=[circle,draw]

\tikzstyle{n} = [->]
\tikzstyle{p} = [<-]
\node[minimum size=1cm,font=\fontsize{15}{96}] (i) at (1,5){$x$};
\node[VR] (1) at (2,1){1};
\node[VR] (2) at (2,3){2};
\node[V] (f) at (0,4){f};
\node[V] (a) at (2,4){a};
\node[V] (e) at (0,2){e};
\node[V] (b) at (2,2){b};
\node[V] (c) at (2,0){c};
\node[VR] (0) at (0,0){0};

\draw[n] (0)--(e) node[midway, left] {0.5}; 

\draw[n] (0.south east)--(c.south west) node[midway, below] {0.5}; 
\draw[p] (0.north east)--(c.north west) node[midway, above] {1}; 

\draw[n] (c.north east)--(1.south east) node[midway, right] {0.5};
\draw[p] (c.north west)--(1.south west) node[midway, left] {1};

\draw[p] (1)--(b) node[midway, right] {0.5};

\draw[p] (b.north west)--(2.south west) node[midway, left] {0.5};
\draw[n] (b.north east)--(2.south east) node[midway, right] {0.5};

\draw[p] (a.south east)--(2.north east) node[midway, right] {0.5};
\draw[n] (a.south west)--(2.north west) node[midway, left] {0.5};

\draw[p] (b)--(e) node[midway, above] {0.5};

\node () at (0,-0.5) {};
\end{tikzpicture}
\begin{tikzpicture}[scale=0.7,every node/.style={scale=0.6,minimum size=20pt}]
\tikzstyle{VR}=[circle,draw,fill=black,text=white]
\tikzstyle{V}=[circle,draw]

\tikzstyle{n} = [->]
\tikzstyle{p} = [<-]
\node[minimum size=1cm,font=\fontsize{15}{96}] (i) at (1,5){$y$};
\node[VR] (1) at (2,1){1};
\node[VR] (2) at (2,3){2};
\node[V] (f) at (0,4){f};
\node[V] (a) at (2,4){a};
\node[V] (e) at (0,2){e};
\node[V] (b) at (2,2){b};
\node[V] (c) at (2,0){c};
\node[VR] (0) at (0,0){0};

\draw[n] (0)--(e) node[midway, left] {1}; 

\draw[n] (0)--(c) node[midway, below] {1}; 

\draw[n] (c)--(1) node[midway, right] {1};

\draw[n] (b)--(2) node[midway, right] {1};

\draw[p] (b)--(e) node[midway, above] {1};

\node () at (0,-0.5) {};
\end{tikzpicture}
}
\end{center}
\caption{Example with $z^{*}_{LP} <  z'^{*}_{LP}$ }
\end{figure}
Solving the linear relaxation of (SCF) leads to $z_{LP}'^* = 20$ because of constraints (10) and (11).
On the other side, we can easily build a feasible solution of the linear relaxation of (SCFS) with a cost $18$ as shown on Figure \ref{fig:feassol}. Thus, in this example we have $z^{*}_{LP} <  z'^{*}_{LP}$.

\end{proof}

In the following, we introduce improvements to offset the weakness of the initial formulation and get closer, or overcome, the quality of the Standard TSP single-commodity flow formulation.

\subsection{Strengthening of the bound\label{sssec:bf}}

Constraints \eqref{eq:SCFS_ybound} are big-M constraints that make the formulation quite weak and can be improved following Letchford, Nasiri and Theis \cite{letchford2013compact}.

Without loss of generality, we can assume that the order picker delivers the unit of commodity due to each required vertex at his first visit. Due to the warehouse structure, a minimum number  $n_R(i)$ of required vertices might have to be visited before visiting vertex $i$ (in particular after preprocessing).  We can apply a shortest path algorithm to compute $n_R(i)$ and reinforce the bound on $y_{ij}$. Constraints \eqref{eq:SCFS_ybound} are replaced by: 
\begin{equation}
y_{ij} \leq \left(n - n_R(i)\right) x_{ij} \ \forall i \in V, j\in \Gamma(i) \label{eq:SCFS_ybound+}\\
\end{equation}

\subsection{Additional cuts}

The Dantzig-Johnson-Fulkerson formulation based on sub-tour elimination constraints is the most well known formulation of the TSP and exhibits a very strong linear relaxation. We consider here a polynomial number of sub-tour elimination constraints depending on the warehouse dimensions rather than the number of products. In the single-flow formulation, the connectivity between all the required vertices is guaranteed. However, because of the big-M constraints \eqref{eq:SCFS_ybound+}, the fractional value of $x$ can be really small. Many sets of vertices are violating the sub-tour elimination constraints in practice. We focus on sets defined as cuts $(S,\bar{S})$ partitioning the warehouse into two subsets $S$, $\bar{S} \subset V$  where $S\cap R \neq \emptyset$ and $\bar{S}\cap R \neq \emptyset$. In other words, each subset contains at least one required vertex. Thus, there must be at least one arc going from $S$ to $\bar{S}$ and at least one arc going from $\bar{S}$ to $S$. The following valid inequality is added for each cut  $(S,\bar{S})$: 
\begin{equation}
x\left(S:\bar{S}\right) \geq 1
\end{equation}

\begin{rem}
It is easy to see, with  constraints of conservation \eqref{eq:SCFS_ass2} that this inequality is sufficient to impose also $x\left(\bar{S}:S\right) \geq 1$
\end{rem}

\begin{figure}
\centering
\includegraphics[scale=0.3]{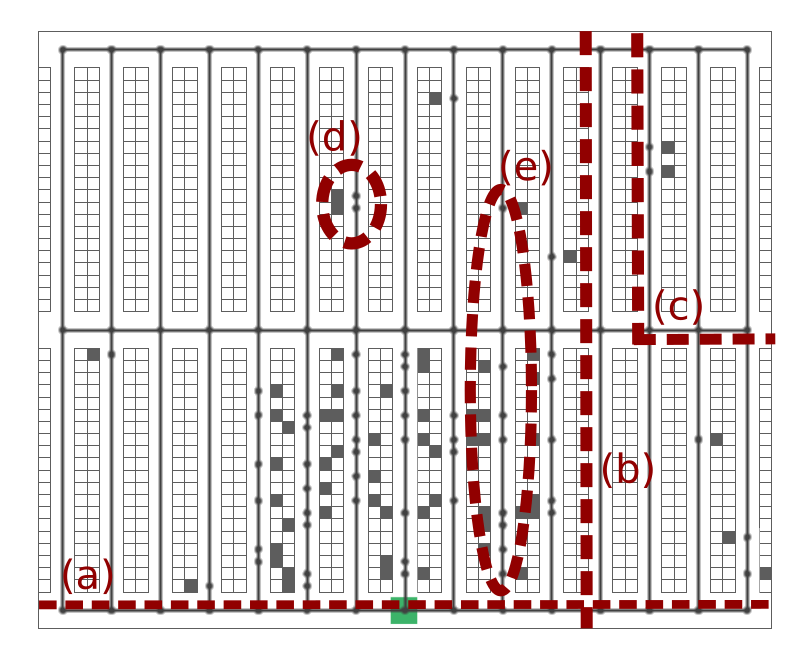}
\caption{Examples of the different cuts in a warehouse \label{fig:cuts}}
\end{figure}
We introduce the following cuts (see Figure \ref{fig:cuts}):
\begin{itemize}
\item Line cuts: horizontal (resp. vertical) cuts $C = (S,\bar{S})$ separating the warehouse horizontally (resp. vertically). See cuts (a) and (b) on Figure \ref{fig:cuts}.
\item Corner boxes cuts: we can combine horizontal and vertical cuts to create boxes attached to a corner of a warehouse. See cut (c) on Figure \ref{fig:cuts}.
\item Sub-aisle connexity:  we define $S \subset V$ as a set of adjacent vertices in the same sub-aisle ($|S| \geq 2$). See cut (d) on Figure \ref{fig:cuts}. There exists at most 6 sets of adjacent vertices in a sub-aisle since there are at most 4  products.
\item Cross cuts: S is defined as the interval between the highest product and the lowest product of two adjacent sub-aisles in the same aisle. See cut~(e) on Figure \ref{fig:cuts}.
\end{itemize}

Note that the total number of cuts considered is polynomial in the size of the warehouse (complexity in $\mathcal{O}(h v) $). We can also add valid inequalities implied by the structure of the graph: 
\begin{itemize}
\item Intersection connexity: a path reaching an intersection (except the depot) can not immediately backtrack in an optimal solution. Namely, if an arc comes in an intersection, i.e, a Steiner vertex, by one side, an arc must go out by \textbf{another} side.
Thus, we add the following constraints:\\
\begin{align*}
x_{ij} \leq \sum\limits_{k \in \Gamma(i)\setminus \{j\}} x_{ki} \quad \forall i \in I , j \in \Gamma(i)  \\
x_{ji} \leq \sum\limits_{k \in \Gamma(i)\setminus \{j\}} x_{ik} \quad \forall i \in I, j \in \Gamma(i)
\end{align*}
\item Patterns: we can identify logical implications from the 6 ways to go through a sub-aisle. We recall that a sub-aisle is surrounded by two intersections (denoted $s$ and $t$) and, after the preprocessing, contains at most 4 products (denoted $a$,$b$,$c$,$d$).

\begin{tabular}{m{8cm}r}
$ $ & \multirow{7}{*}{
\begin{tikzpicture}
\node[draw,circle,scale=0.3] (a) at (0,2) {};
\node[draw,circle,fill=black,text=white,scale=0.3] (a) at (0,1.5) {};
\node[draw,circle,fill=black,text=white,scale=0.3] (a) at (0,1.2) {};
\node[draw,circle,fill=black,text=white,scale=0.3] (a) at (0,0.8) {};
\node[draw,circle,fill=black,text=white,scale=0.3] (a) at (0,0.5) {};
\node[draw,circle,scale=0.3] (a) at (0,0) {};
\node[scale=0.6]  (a) at (0.3,2) {s};
\node[scale=0.6] (a) at (0.3,1.5) {d};
\node[scale=0.6] (a) at (0.3,1.2) {c};
\node[scale=0.6] (a) at (0.3,0.8) {b};
\node[scale=0.6] (a) at (0.3,0.5) {a};
\node[scale=0.6] (a) at (0.3,0) {t};
\end{tikzpicture}} \\
$x_{sd} \Rightarrow x_{dc} \text{ and } x_{ds} \Rightarrow x_{cd} $
&  \\
$ $ & \\
$x_{ta} \Rightarrow x_{ab} \text{ and } x_{at} \Rightarrow x_{dc}$ &  \\
$ $ & \\
$x_{cb} \Rightarrow x_{dc} \wedge x_{ba} \text{ and } x_{bc} \Rightarrow x_{ab} \wedge x_{cd}$ & \\
$ $ & 
\end{tabular}

 They are added as linear constraints thanks to the following scheme:
 $ P \Rightarrow Q $ is linearized into $P\leq Q$ and
$P \Rightarrow Q \wedge R $ is linearized into $z \geq Q + R - 1$, $z \leq Q$, $z \leq R$ and $P\leq z$, where $P$, $Q$, $R$ and $z$ are boolean variables.
\end{itemize}

Finally, notice that any tour leads to a symmetric one by reversing the direction of traversal. Although constraints can be added to break this symmetry, it did not pay off in our experiments and they were removed.

\section{Dynamic programming \label{sec:pdyn}}

A dynamic programming algorithm has been proposed by Ratliff and Rosenthal for the picking problem in a rectangular warehouse with two cross-aisles in 1983 \cite{ratliff1983order}. It was extended in 2001 to the case of three cross-aisles by Roodbergen and De Koster \cite{roodbergen2001routingDP}. More generally, the rectilinear TSP can be solved by dynamic programming using the very same ideas. An algorithm, proposed by Cambazard and Catusse \cite{cambazard2015fixed}, is proved to have a $O(hn7^h)$ (or more precisely $O(hv7^h)$) runtime complexity where $n$ cities are located on $h$ horizontal lines and $v$ vertical lines. The distance considered between any pair of cities is the $l_1$ (rectilinear or manhattan) distance. This algorithm is directly applicable to the picking problem which can be seen as a specific case. We will now give the key ideas and a summary of this approach in the present section. We refer the reader to \cite{cambazard2015fixed} for the details and in particular the proofs of correctness and complexity analysis.

We consider the problem in an undirected grid graph such as the one shown in Figure \ref{fig:graph}. The set of vertices located on a vertical aisle or two adjacent vertical aisles is a planar \textbf{separator} of this grid graph (e.g. $\{5,9,14\}$, $\{d,9,14\}$ or $\{d,10,14\}$ are separators in the graph of Figure \ref{fig:graph}). The rationale of the dynamic programming algorithm is that the problem can be split into two sub-problems, to the right and to the left of a separator by considering all the possible configurations of the separator (degree parity of the vertices and connected components described below). The algorithm builds a \textbf{tour subgraph} that can be directed as a post-processing step to obtain a picking tour. An example of such tour subgraph is shown on Figure \ref{toursbg}.

\begin{figure}[b!]
\centering
\includegraphics[scale=0.8]{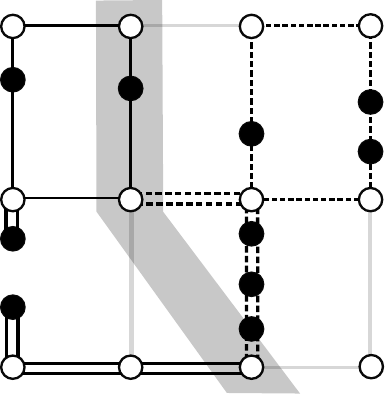}
 \label{fig:partial_graph}
\caption{The black edges represent a partial tour subgraph and the dashed edges represent one possible completion to a complete tour subgraph. The white vertices in the gray area represents the current state \{(E,E,E)(1,1,2)\}. \label{toursbg}}
\end{figure}

\subsection{States}
A state of the dynamic program is a possible configuration of a separator. In the following, such a state $\omega$ is denoted $\omega = \{(x_1, \ldots, x_h),(c_1, \ldots, c_h)\}$ where $x_i \in \{U,E,0\}$ and $c_i$ are respectively the parity label and the connected component of the $i$-th vertex of $\omega$. We use the same notation as Ratliff and Rosenthal \cite{ratliff1983order} to describe degree parities: even = E, odd = U (uneven) and zero = 0. Connected components are described by their indices or "$-$" for a zero degree. Figure \ref{toursbg} gives an example of a state where all vertices have an even degree and belong to two distinct connected components.
 
 \subsection{Transitions}

There are two types of transitions between states: vertical and horizontal transitions, corresponding to the decisions made on vertical or horizontal edges of the grid graph. Horizontal transitions are of three kinds: no edge, a single edge or a double edge. Vertical transitions are of the six kinds identified by Ratliff and Rosenthal \cite{ratliff1983order}: no edge, a single edge, a single double edge, two double edges connected to the top vertex, two double edges connected to the bottom vertex and four double edges defined by the largest gap (see Figure \ref{fig:6ways}).
\begin{algorithm}[h!]
\begin{algorithmic}[1]
\STATE $\omega_0 \leftarrow \{(0,\ldots,0), (-,\ldots,-)\}$; $T(w_0, 0) = 0$; $Layer_0  \leftarrow \{w_0\}$
\STATE $l \leftarrow 0$
\FOR{each edge $e$ of the grid graph from bottom to top and left to right} \label{line4}
		\STATE $Layer_{l+1} \leftarrow \emptyset$
		\FOR{each state $\omega \in Layer_l$} \label{line5}			
			\FOR{each possible transitions $tr$ for $e$}  \label{line6}
				\STATE $\omega' \leftarrow \omega + tr$  \label{line7}
		\IF{check($\omega', l+1$)} \label{lcheck}
					\IF{$\omega' \in Layer_{l+1}$}
						\IF{$T(\omega', l+1) > T(\omega, l) + length(tr)$} \label{line10}
							\STATE $T(\omega', l+1)\leftarrow T(\omega, l) + length(tr)$
						\ENDIF
					\ELSE
						\STATE $T(\omega', l+1)\leftarrow T(\omega, l) + length(tr)$
						\STATE $Layer_{l+1} \leftarrow Layer_{l+1} \cup \{\omega'\}$  \label{line14}
					\ENDIF
				\ENDIF
			\ENDFOR
		\ENDFOR
		\STATE $l \leftarrow l+1$
\ENDFOR
\STATE $w_{opt} \leftarrow \argmin_{\omega \in L_{hv}} T(\omega, hv)$
\RETURN $w_{opt}$
\caption{Dynamic Programming algorithm for  rectangular picking}
\label{algogo}
\end{algorithmic}
\end{algorithm}

\subsection{Outline of the algorithm}

Algorithm \ref{algogo} processes the edges of the grid graph from bottom to top and then from left to right (line \ref{line4}). Typically on Figure \ref{fig:graph}, the edges would be considered in the following order: (4,8), (8,13), (4,5), (8,9), (13,14), (5,9)  and so on.  All the states obtained after adding $l$ transitions (denoted $Layer_l$) belong to the l-th layer and the algorithm can be seen as a shortest path algorithm in a layered graph (see Figure \ref{fig:gpdyn}). From any state (line \ref{line5}), the possible transitions are considered (line \ref{line6}); three or six depending if the edge considered is a vertical or horizontal one. We denote by $T(\omega, l)$ the value of the shortest path to reach state $\omega$ located on layer $l$. Lines \ref{line7}-\ref{line14} update the possible states of the next layer (l+1) by extending the considered state $\omega$ of layer $l$ with the considered transition $tr$. The new state $\omega'$ might not be a valid tour subgraph and it is checked line \ref{lcheck}. For instance, a partial tour subgraph is not valid if a vertex is left with an odd number of incident edges, if a product is not collected (zero degree) or if it has more than one connected component on the last layer (since the final tour subgraph must be connected). A shortest partial tour subgraph might be already known to reach $\omega'$ and this is checked line \ref{line10}.

An illustrative execution of the algorithm is given Figure~\ref{fig:gpdyn}, where a particular path is outlined showing the relation between states and partial tour subgraphs.

\begin{figure}
\centering
\includegraphics[scale=0.9]{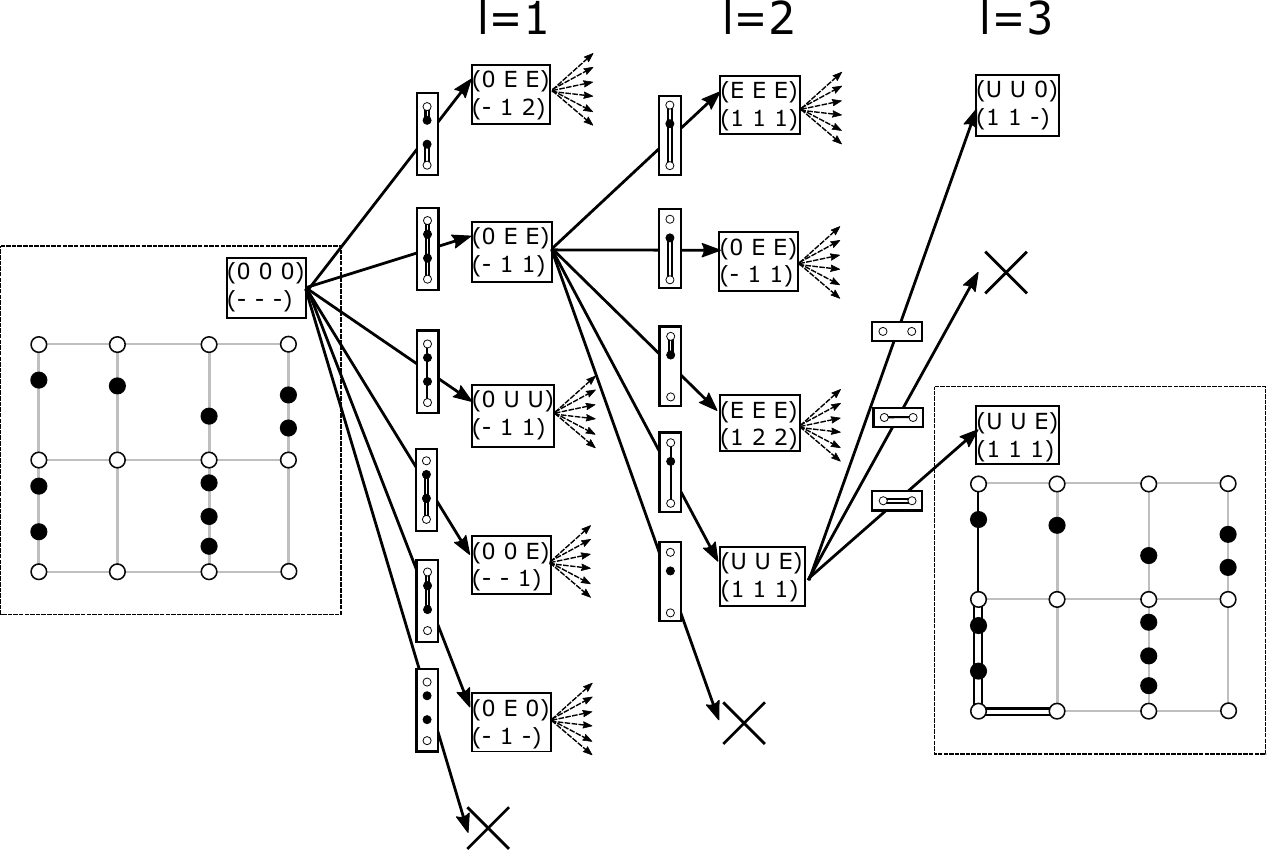}

\caption{Example of the graph underlying the dynamic programming algorithm. Each layer is identified with a value of l. Three transitions are possible from each state. The partial toursubgraph obtained by following the black path is shown on the bottom right corner. \label{fig:gpdyn}}
\end{figure}

\section{Experimental results \label{sec:results}}

\subsection{Implementation}

We used the CPLEX Java API (version 12.6) to solve the different linear programs.
An initial upper bound is given as "warm start" and obtained with the Lin-Kernighan heuristic \cite{helsgaun2000effective} implemented in the software LKH (freely available at \url{http://webhotel4.ruc.dk/~keld/research/LKH/}). This gives a better comparison with the TSP solver Concorde that takes advantage of an initial accurate upper bound\cite{applegate2003chained}. 
The instances solved come from an academic benchmark proposed by Theys, Dullaert and Herroelen \cite{theys2007routing}. Some of these instances seem bigger than realistic data, but allow to test the limits of the algorithms and to perform a comparison with previously published results on the exact same instances \cite{theys2010using}.
Experiments were performed on an Intel Xeon E5-2440 v2 @ 1.9 GHz processor and 32 GB of RAM. The experiments ran with a memory limit of 8 GB of RAM.

\subsection{Description of the instances}

The benchmark of Theys\textit{ et al.} contains 108 classes of instances, where a class is defined by 5 parameters described in Table \ref{tab:instances}. The storage policy is either random-based or volume-based. A random policy means the products are randomly affected in the warehouse. A volume-based policy means a twenty percent of the most demanded items are located near the first cross-aisle. For more details on the instance, see Theys \textit{et al.} \cite{theys2010using}.

\begin{table}
\centering
\begin{tabular}{|c|c|}
\hline
parameter & values\\
\hline
the number of vertical aisles & \{ 5, 15, 60 \} \\
the number of cross-aisles& \{ 3, 6, 11 \} \\
the number of products in the picking list & \{ 15, 60, 240 \} \\
the storage policy& \{ volume (V), random (R) \} \\
the location of the depot in the warehouse& \{ central, decentral \} \\
\hline
\end{tabular}
\caption{Parameters and values of the instances \label{tab:instances}}
\end{table}

Early experiments showed that the location of the depot (central or decentral) does not affect the efficiency of our models. Thus, we report our results on the 54 classes with a central depot where each class contains 10 instances. 


\subsection{Results}

The performances of the following algorithms are compared:
\begin{itemize}
\setlength\itemsep{0.2em}
\item SCFS: the basic Steiner single commodity flow formulation 
\item SCFS\_PP: SCFS with preprocessing of section \ref{sec:pp}
\item SCFS+: SCFS with preprocessing and the additional valid inequalities
\item SCF+: the standard single commodity flow formulation with vertex preprocessing \ref{sec:ppbasic}
\item CDE: Concorde 
\item CDE+: Concorde with preprocessed input (described in section \ref{sec:ppconcorde})
\item PDYN: dynamic programming (section \ref{sec:pdyn})
\end{itemize}

To start with, we report an analysis of the problem's size, which is the biggest advantage of the SCFS formulation. Then, we look at the strength of the linear relaxation of each formulation and note that SCFS  strongly benefits from the improvements proposed. In the end, we compare resolution times. The tables report the average value of the quantity studied (size, gap, cpu times) for all instances restricted to the value of the parameter given as the column header.

\paragraph{Size analysis}

\begin{table}
\centering
\footnotesize
\begin{adjustwidth}{-.5in}{-.5in}  
\begin{tabularx}{1.3\textwidth}{|c|c| *{11}{|>{\centering\arraybackslash}X}|}
\hline
& \multirow{2}{*}{Total} & \multicolumn{2}{c|}{Storage policy} & \multicolumn{3}{c|}{\# aisles} & \multicolumn{3}{c|}{\# cross-aisles} & \multicolumn{3}{c|}{\# products} \\
& & R & V & 5 & 15 & 60 & 3 & 6 & 11 & 15 & 60 & 240 \\
\hline
\hline
TSP &  20580 & 20580 & 20580 & 20580 & 20580 &20580 & 20580 & 20580 & 20580 & 240 & 3660 & 57840 \\
 \hdashline
TSP+ & 6676 & 9228 & 4123 & 2319 & 6042 & 11665 & 3744 & 7238 & 9046 & 217 & 2350 & 17460 \\
\hline
Evolution & -68\%  & -55\%  & -80\%  & -89\%  & -71\%  & -43\%  & -82\%  & -65\%  & -56\%  & -9\%  & -36\%  & -70\%  \\
\hline
\hline
Steiner & 768 & 790 & 746 &  192 & 481 & 1633 & 355 & 705 & 1245 & 678 & 743 & 884  \\
 \hdashline
Steiner+ & 474 & 543  & 405 &  157 & 347 & 917 & 267 & 456 & 698 & 216 & 446 & 760\\
\hline
Evolution & -38\% & -31\% & -46\% & -18\% & -28\% & -44\% & -25\% & -35\% & -44\% & -68\% & -40\% & -14\% \\
\hline
\end{tabularx}
\caption{Average number of arcs in TSP graph and Steiner graph, with and without preprocessing\label{tab:size}}
\end{adjustwidth}
\end{table}

%

\indent Table \ref{tab:size} shows the number of arcs in the TSP case (complete graph) and in the Steiner graph with and without the preprocessing \ref{sec:pp}.  The lines "Evolution" show the percentage of arcs removed when the preprocessing is applied.
We notice that the preprocessing is really more efficient in the TSP case due to the completeness of the graph. We also notice that the number of aisles, cross-aisles and products have an opposite effect depending on the solver.

Finally note that the number of arcs in the Steiner graph is much smaller than in the complete graph which significantly improve the memory scaling of the MILP formulations.

\paragraph{Analysis of lower bounds}

\begin{table}
\centering
\footnotesize
\begin{adjustwidth}{-.5in}{-.5in}  
\begin{tabularx}{1.3\textwidth}{|c|c| *{11}{|>{\centering\arraybackslash}X}|}
\hline
& \multirow{2}{*}{Total} & \multicolumn{2}{c|}{Storage policy} & \multicolumn{3}{c|}{\# aisles} & \multicolumn{3}{c|}{\# cross-aisles} & \multicolumn{3}{c|}{\# products} \\
& & R & V & 5 & 15 & 60 & 3 & 6 & 11 & 15 & 60 & 240 \\
\hline
\hline
SCFS & 46.4\% & 43.12\% & 49.68\% & 34.85\% & 45.8\% & 58.53\% & 41.95\% & 47.13\% & 50.11\% & 46.9\% & 45.24\% & 47.05\% \\
\hline
SCFS\_PP & 39.32\% & 37.41\% & 41.22\% & 24.49\% & 38.91\% & 54.55\% & 33.23\% & 40.68\% & 44.03\% & 44.74\% & 40.55\% & 32.65\% \\
\hline
SCFS+ & 1.4\% & 1.68\% & 1.13\% & 2.7\% & 1.05\% & 0.46\% & 1.72\% & 1.28\% & 1.21\% & 0.4\% & 1.21\% & 2.6\% \\
\hline
\end{tabularx}
\caption{Average gap of linear relaxations for single-commodity flow formulations\label{tab:scfs_comp2}}
\end{adjustwidth}
\end{table}
\begin{table}
\centering
\footnotesize
\begin{adjustwidth}{-.5in}{-.5in}  
\begin{tabularx}{1.3\textwidth}{|c|c| *{11}{|>{\centering\arraybackslash}X}|}
\hline
& \multirow{2}{*}{Total} & \multicolumn{2}{c|}{Storage policy} & \multicolumn{3}{c|}{\# aisles} & \multicolumn{3}{c|}{\# cross-aisles} & \multicolumn{3}{c|}{\# products} \\
& & R & V & 5 & 15 & 60 & 3 & 6 & 11 & 15 & 60 & 240 \\
\hline
SCFS+ &  0.65\% & 1.01\% & 0.29\% & 1.33\% & 0.36\% & 0.27\% & 0.81\% & 0.53\% & 0.61\% & 0.04\% & 0.34\% & 1.59\%\\
\hline
SCF+ & 6.44\% & 7.47\% & 5.48\% & 2.39\% & 4.73\% & 13.3\% & 7.6\% & 6.27\% & 5.34\% & 0.37\% & 5.09\% & 15.36\%   \\
\hline
CDE+ & \multicolumn{12}{c|}{ < 0.1\% } \\
\hline
\end{tabularx}
\caption{Average gap of root node relaxations for single-commodity flow formulations and Concorde \label{tab:scfs_comp21}}
\end{adjustwidth}
\end{table}

%

\indent Table \ref{tab:scfs_comp2} compares the average gap between the \textbf{linear relaxation} and the optimal value for each parameter and for the different single-commodity flow formulations. The gap is computed as $\frac{z^{*} - z^*_{LP}}{z^*} \times 100 $. The results clearly demonstrate that the improvements proposed are very effective to strengthen the linear relaxation. The gap moves from  46\% to 1\% in average without increasing significantly the computing time which moves from 0.07 second to 0.2 second in average.
To compare our results to Concorde (see Table \ref{tab:scfs_comp21}), we observe the gap between the \textbf{lower bound at the root node} of the search tree and the optimal value. This lower bound is better than the linear relaxation since CPLEX and Concorde apply many techniques to improve it. 

%
%

SCFS+ has a strong linear relaxation despite the fact that the initial formulation SCFS is weaker than all the other ones. In practice, we observe that the linear relaxation of the improved Steiner single commodity flow formulation (SCFS+) is significantly \textbf{stronger} than the linear relaxation of the improved standard single commodity flow formulation (SCF+).

\paragraph{Performances}
\indent We set a time-limit of 30 minutes for SCFS+ and SCF+. Some instances were unsolved in this time limit. Table \ref{tab:bad_comp} shows the number of unsolved instances after 30 minutes of processing depending on the different parameters.


\begin{table}
\centering
\footnotesize
\begin{adjustwidth}{-.5in}{-.5in}  
\begin{tabularx}{1.3\textwidth}{|c|c| *{11}{|>{\centering\arraybackslash}X}|}
\hline
& \multirow{2}{*}{Total} & \multicolumn{2}{c|}{Storage policy} & \multicolumn{3}{c|}{\# aisles} & \multicolumn{3}{c|}{\# cross-aisles} & \multicolumn{3}{c|}{\# products} \\
& & R & V & 5 & 15 & 60 & 3 & 6 & 11 & 15 & 60 & 240 \\
\hline
SCFS+ & 18 & 18 & 0 & 1 & 4 & 13 & 1 & 2 & 15 & 0 & 0 & 18 \\
\hline
SCF+ & 136 & 88 & 48 & 19 & 34 & 83 & 51 & 41 & 44 & 0 & 26 & 110  \\
\hline
PDYN & 180 & 90 & 90 & 60 & 60 & 60 & 0 & 0 & 180 & 60 & 60 & 60 \\ 
\hline 
\hline
\# instances &  540 & 270 & 270 & 180 & 180 & 180 & 180 & 180 & 180 & 180 &180 &180\\
\hline
\end{tabularx}
\caption{Number of unsolved instances after 30 minutes\label{tab:bad_comp}}
\end{adjustwidth}
\end{table}

\begin{table}
\centering
\footnotesize
\begin{adjustwidth}{-.5in}{-.5in}  
\begin{tabularx}{1.3\textwidth}{|c|c| *{11}{|>{\centering\arraybackslash}X}|}
\hline
& \multirow{2}{*}{Total} & \multicolumn{2}{c|}{Storage policy} & \multicolumn{3}{c|}{\# aisles} & \multicolumn{3}{c|}{\# cross-aisles} & \multicolumn{3}{c|}{\# products} \\
& & R & V & 5 & 15 & 60 & 3 & 6 & 11 & 15 & 60 & 240 \\
\hline
SCFS+ &  1.3\% & 1.3\% & - & 0.85\% & 0.46\% & 1.59\% & 0.23\% & 0.29\% & 1.51\% & - & - & 1.3\%\\
\hline
SCF+ & 19.63\% & 19.1\% & 20.51\% & 10.19\% & 18.8\% & 23.45\% & 16.37\% & 21.19\% & 23.91\% & - & 7.79\% & 23.47\%   \\
\hline
\end{tabularx}
\caption{Average gap between best upper and lower bounds for unsolved instances \label{tab:scfs_badgap}}
\end{adjustwidth}
\end{table}

The improved Steiner formulation completely outperforms the standard compact TSP formulation. It still cannot solve some instances but 
they are from only 5 classes from the biggest ones and only with a random policy (15\_11\_240\_R, 5\_11\_240\_R 60\_11\_240\_R, 60\_3\_240\_R, 60\_6\_240\_R). On the other hand, instances from 16 classes cannot be solved by the SCF solver. The only parameter guaranteeing an optimal resolution is \# products = 15.
Moreover, as the Table \ref{tab:scfs_badgap} shows, the gap between the lower and upper bounds is really smaller for the SCFS+, which indicates that the upper bound can be used as a good feasible solution.\\
The dynamic program solves all the instances with less than 11 cross-aisles. This is a known limitation of this approach since the algorithm has an exponential complexity in the number of cross-aisles.

%

 To compare resolution time, we choose the instances solved in less than 30 minutes by solver SCFS+. Formulation SCF+ has too many unsolved instances after this time limit so it appears irrelevant to include it in the comparison.  Table \ref{tab:time_comp} shows the numerical results. For each parameter, SCFS+ is slower than Concorde and dynamic programming. However, on many instances the computing time is reasonable.\\
We also included results for solver CDE to show the impact of the preprocessing on the Concorde solver. Note that, with the preprocessing, any instance of the entire benchmark of \cite{theys2010using} can be solved optimally in less than 1 minute by Concorde.\\
The dynamic programming approach is the quickest since all the accepted instances are solved in less than a second.

\begin{table}
\centering
\footnotesize
\begin{adjustwidth}{-.5in}{-.5in}  
\begin{tabularx}{1.3\textwidth}{|c|c| *{11}{|>{\centering\arraybackslash}X}|}
\hline
& \multirow{2}{*}{Total} & \multicolumn{2}{c|}{Storage policy} & \multicolumn{3}{c|}{\# aisles} & \multicolumn{3}{c|}{\# cross-aisles} & \multicolumn{3}{c|}{\# products} \\
& & R & V & 5 & 15 & 60 & 3 & 6 & 11 & 15 & 60 & 240 \\
\hline
SCFS+ & 36.07 & 61.21 & 12.62 & 23.89 & 56.96 & 27.12 & 3.44 & 35.88 & 71.69 & 0.07 & 4.05 & 111.64\\
\hline
PDYN & 0.27 & 0.28 & 0.27 & 0.05 & 0.16 & 0.61 & 0 & 0.54 & - & 0.24 & 0.27 & 0.30 \\
\hline
CDE & 6.86 & 12.8 & 0.93 & 17.61 & 2.11 & 0.88 &   14.82 & 3.89 & 1.88 & 0.01 & 0.13 & 20.45 \\
\hline
CDE+ & 1.60 & 2.97 & 0.23 & 3.45 & 1.04 & 0.30 & 2.20 & 1.55 & 1.04 & <0.01 & 0.1 & 4.68\\
\hline
\end{tabularx}
\caption{Average time of optimal resolution (in seconds) for instances solved in less than 30 minutes with each solver \label{tab:time_comp}}
\end{adjustwidth}
\end{table}

%
%

\paragraph{Instances from Scholz \textit{et al}}
Scholz \textit{et al.} introduced other instances for the case of a single-block layout \cite{scholz2016new}. 
\begin{table}[h!]
\centering
\begin{tabular}{|c|c|}
\hline
parameter & values\\
 \hline 
  number of products in a sub-aisle & 45\\
  number of cross-aisles& 2\\
 number of vertical aisles & \{ 5, 10, 15, 20, 25, 30 \} \\
 number of products in the picking list & \{ 30, 45, 60, 75, 90 \} \\
location of the depot in the warehouse&  decentral \\
\hline
\end{tabular}
\caption{Parameters and values of the Scholz instances \label{tab:Sinstances}}
\end{table}

 The parameters are given in table \ref{tab:Sinstances} and form 30 classes of instances. For each one of their classes, we generated 10 instances of the same size. As expected, since these instances have only one block, the dynamic program is extremely efficient and solves any instance optimally in less than a second. Our MILP model also provides optimal solution for all instances almost instantaneously. Indeed, the maximal time of resolution is 6.7 seconds for the biggest instance (30 aisles and 90 products) and the average on all the instances is 0.28 seconds. Moreover, the model introduced by Scholz \textit{et al.} fails to find the optimal solution for 44 instances. 
 Table \ref{tab:Stime_comp} details the computation time to find the optimal solution for these solvers. Please note that our instances are not exactly the same than those solved by Scholz \textit{et al.} but were generated with the same parameters.

\begin{table}
\centering
\footnotesize
\begin{adjustwidth}{-.5in}{-.5in}  
\begin{tabularx}{1.3\textwidth}{|c|c| *{14}{|>{\centering\arraybackslash}X}|}
\hline
& \multirow{2}{*}{Total}  & \multicolumn{6}{c|}{\# aisles}  & \multicolumn{5}{c|}{\# products} \\
&  & 5 &10 & 15  & 20 & 25 & 30 & 30 & 45 & 60 & 75 & 90 \\
\hline
Scholz  & 167.7 & 0.11 & 1.21 & 10.54 & 170.91 & 314.56 & 508.87 & 45.53 & 87.56 & 149.03 & 233.24 & 323.14  \\
\hline
\hline
SCFS+ & 0.28  &  0.06 & 0.09 & 0.13 & 0.34 & 0.54 & 0.54 & 0.07 & 0.13 & 0.25 & 0.38 & 0.6 \\
\hline
PDYN  & $<$ 0.001  & \multicolumn{11}{c|}{$<$ 0.001}  \\
\hline
\end{tabularx}
\caption{Average time of optimal resolution (in seconds) for "Scholz instances", solved in less than 30 minutes with each solver \label{tab:Stime_comp}}
\end{adjustwidth}
\end{table}

\paragraph{Discussion}

A number of fast heuristics have been proposed in the past to
solve the order picking problem without any
guarantee of optimality. And some of these heuristics have been
improved to achieve better solution at the
expanse of their running time\cite{theys2010using}. Note that the dynamic program proposed
in this paper can solve instances up to 6 cross-aisles with an average
time of 0.27 seconds \textit{i.e.}, the same order of running time reported in
Theys et al (between 0.34 and 0.43 seconds for 3 and 6 cross-aisles) for advanced heuristics.
We also showed that the use of
Concorde with the preprocessing we propose in this paper gives
excellent results even beyond 6 cross-aisles. 
 In average, the computation time is less than 2 seconds to compute an exact solution and this method can be applied on instances of any size.
 The use of heuristics  is
therefore not relevant in our opinion if we are only concerned with
minimizing the distance of the tour, since the dynamic program or Concorde with preprocessing provide an exact solution in a really small amount of time. However, these methods share a weakness with the dedicated heuristics: they cannot easily accommodate side-constraints. Our MILP model is not as efficient as the other methods but it outperforms previously proposed MILP and solve many instances in a reasonable amount of time. Thus, it is relevant to employ it when the user wants an exact solution and has specific requirements that can be modeled with linear constraints.

%
%
\section{Conclusion \label{sec:conclusion}}

We have studied two exact algorithms for the picking problem based on   dynamic programming and mixed integer linear programming.\\
The first approach was previously proposed for warehouses with up to three cross-aisles. We extend it to any number of cross-aisles which has often been mentioned but never done before. This algorithm proves to be extremely efficient for realistic size of warehouses. However, it can not accommodate side constraints such as precedences, flow directions or multiple depots: a MILP is better suited to deal with these requirements.\\
With this in mind, we showed that, on one hand, the compact formulations based on modeling the problem as a TSP do not scale in memory and are unable to solve realistic size instances. On the other hand, a flow based formulation modeling the problem as a Steiner TSP is very sparse but has a weak linear relaxation. As a result, it is also inefficient in practice.\\
Scholz \textit{et al.} proposed a new formulation, with improvement by preprocessing. However, their model is rather complex and does not yet provide convincing results regarding its efficiency when the layout grows\cite{scholz2016new}.
We thus propose a number of improvements by taking advantage of the warehouse structure. These improvements are based on valid inequalities and procedures to significantly reduce the instance's size without loosing optimality. The resulting model remains sparse and exhibits a strong linear relaxation in practice. It outperforms the compact TSP model and solves very large instances almost as efficiently as dedicated TSP approaches on the benchmark studied.

%

%
%

Note finally that some of the ideas proposed here can be applied to improve the efficiency of Concorde. The entire benchmark proposed by Theys, Br{\"a}ysy, Dullaert and Raa \cite{theys2010using} can thus be solved to optimality very efficiently without the need of the heuristics proposed by the same authors.

The analysis of the results showed that the improvement was stronger when the products were stored with a volume policy. It may be a promising track to solve the picking problem jointly with other warehouse issues such as storage policy or batching. Valle \textit{et al.} follow this track in a recent paper \cite{valle2016modelling} as did Won and Olafsson a few years ago\cite{won2005joint}.

\bibliographystyle{plain}

\end{document}